\documentclass[a4paper,showkeys,amsmath,amssymb,superscriptaddress,prx,twocolumn,floatfix,showpacs]{revtex4}

\usepackage{amsopn}
\usepackage{amsmath}
\usepackage{amssymb}
\usepackage{amsthm}
\usepackage{hyperref}
\usepackage{graphicx}
\usepackage{color}
\usepackage{listings}
\usepackage{subfig}
\usepackage{enumerate}
\usepackage{todonotes}

\newcommand{\ie}{{\emph{i.e.\/}}}

\providecommand{\keywords}[1]{\newline Keywords: #1}

\providecommand{\pacs}[1]{\newline PACS: #1}

\newtheorem{definition}{Definition}
\newtheorem{remark}{Remark}

\DeclareMathOperator{\tr}{tr}
\DeclareMathOperator{\Tr}{Tr}

\DeclareMathOperator{\diag}{diag}

\newcommand{\R}{\ensuremath{\mathbb{R}}}

\newcommand{\C}{\ensuremath{\mathbb{C}}}

\newcommand{\ket}[1]{\ensuremath{|#1\rangle}}
\newcommand{\bra}[1]{\ensuremath{\langle#1|}}
\newcommand{\ketbra}[2]{\ensuremath{\ket{#1}\bra{#2}}}
\newcommand{\proj}[1]{\ensuremath{\ketbra{#1}{#1}}}
\newcommand{\braket}[2]{\ensuremath{\langle{#1}|{#2}\rangle}}

\newcommand{\1}{{\rm 1\hspace{-0.9mm}l}}

\newcommand{\ii}{\ensuremath{\mathrm{i}}}

\newcommand{\TT}{\mathcal{T}}
\newcommand{\UU}{\mathcal{U}}
\newcommand{\HH}{\mathcal{H}}
\newcommand{\PP}{\mathcal{P}}

\newcommand{\diaguni}{\ensuremath{\mathcal{DU}}}
\newcommand{\diagmodul}{\ensuremath{\mathcal{D}^{\leq 1}_d}}

\newcommand{\ketV}[1]{\ensuremath{|#1\rangle\!\rangle}}
\newcommand{\braV}[1]{\ensuremath{\langle\!\langle#1|}}
\newcommand{\ketbraV}[2]{\ensuremath{\ketV{#1}\braV{#2}}}
\newcommand{\projV}[1]{\ensuremath{\ketbraV{#1}{#1}}}

\newtheorem{lemma}{Lemma}
\newtheorem{theorem}{Theorem}

\newtheorem{corollary}{Corollary}
\newtheorem{proposition}{Proposition}

\renewcommand{\Re}{\mathrm{Re}}

\usepackage{caption}
\def\>{\rangle}
\def\<{\langle}
\begin{document}
\author{Zbigniew Pucha{\l}a}
\affiliation{Institute of Theoretical and Applied Informatics, Polish Academy
of Sciences, ulica Ba{\l}tycka 5, 44-100 Gliwice, Poland}
\affiliation{
Faculty of Physics, Astronomy and Applied Computer Science, Jagiellonian 
University, ulica Stanis\l{}awa \L{}ojasiewicza 11, 30-348 Krak\'o{}w, Poland}

\author{{\L}ukasz Pawela\footnote{Corresponding author, E-mail:
lpawela@iitis.pl}}
\affiliation{Institute of Theoretical and Applied Informatics, Polish Academy
of Sciences, ulica Ba{\l}tycka 5, 44-100 Gliwice, Poland}
\affiliation{Institute of Informatics, National Quantum Information Centre,
Faculty of Mathematics, Physics and Informatics, University of Gda{\'n}sk, ul.
Wita Stwosza~57, 80-308 Gda{\'n}sk, Poland}

\author{Aleksandra Krawiec}
\affiliation{Institute of Theoretical and Applied Informatics, Polish Academy
of Sciences, ulica Ba{\l}tycka 5, 44-100 Gliwice, Poland}
\affiliation{Institute of Mathematics, University of Silesia, ul. Bankowa 14, 
40-007 Katowice, Poland}

\author{Ryszard Kukulski}
\affiliation{Institute of Theoretical and Applied Informatics, Polish Academy
of Sciences, ulica Ba{\l}tycka 5, 44-100 Gliwice, Poland}
\affiliation{Institute of Mathematics, University of Silesia, ul. Bankowa 14, 
40-007 Katowice, Poland}

\title[Single-shot discrimination of quantum measurements]{Strategies for 
optimal single-shot discrimination of quantum measurements}

\date{April 24, 2018}

\begin{abstract}
In this work we study the problem of single-shot discrimination of von Neumann
measurements, which we associate with measure-and-prepare channels. There are
two possible approaches to this problem. The first one is simple and does not
utilize entanglement. We focus only on the discrimination of classical
probability distributions, which are outputs of the channels. We find necessary
and sufficient criterion for perfect discrimination in this case. A more
advanced approach requires the usage of entanglement. We quantify the distance
between two measurements in terms of the diamond norm (called sometimes the
completely bounded trace norm). We provide an exact expression for the optimal
probability of correct distinction and relate it to the discrimination of
unitary channels. We also state a necessary and sufficient condition for perfect
discrimination and a semidefinite program which checks this condition. Our main
result, however, is a cone program which calculates the distance between the
measurements and hence provides an upper bound on the probability of their
correct distinction. As a by-product, the program finds a strategy (input state)
which achieves this bound. Finally,  we provide a full description for the cases
of Fourier matrices and mirror isometries.
\end{abstract}


\maketitle

\section{Introduction}

The state of a quantum system is inherently non-observable. Despite this,
quantum states have been the focus of quantum theory since its beginning as they
provide a way of computing the value of any observable. The picture changes when
we consider two quantum states and ask about their distance. This quantity can,
in principle, be measured, and provides an upper bound on the probability of
discriminating between these states. The latter was shown by in
Helstrom~\cite{helstrom1976quantum}. Such problems are fundamental in quantum
information science and quantum physics, and have attracted a lot of attention
in recent years. These range from experimental studies
\cite{mosley2006experimental,clarke2001experimental,mohseni2004optical},
theoretical considerations of finite-dimensional random quantum
states~\cite{mejia2016difference} to asymptotic properties of random quantum
states~\cite{puchala2016distinguishability}. This approach can be extended to
quantum channels via the Choi-Jamio{\l}kowski
isomorphism~\cite{choi1975completely,jamiolkowski1972linear}. Helstrom's result
can be easily extended to this case and once again we obtain a simple expression
for the upper bound for the probability of discriminating two quantum channels.
There is, however, one additional feature in this setting, which is the input
state. This input state is what we call ``the strategy'' for discriminating
quantum channels. Due to the complicated structure of the set of quantum
channels, the problem has been studied in the limit of large input and output
dimensions~\cite{nechita2016almost}. In this paper we focus on the problem of
discriminating quantum measurements which are viewed as a subclass of quantum
channels.

The problem of discriminating quantum measurements is of the utmost importance
in modern quantum information science. Imagine we have an unknown measurement 
device, a black-box. The only information we have is that it performs one of 
two measurements, say $\mathcal{S}$ and $\TT$. Our goal is two-fold. First, we 
want to tell whether it is possible to discriminate $\mathcal{S}$ and $\TT$ 
perfectly, \ie\ with probability equal to one. If this is not the case, we 
would like to know the upper bound of such a probability. Secondly, we need to 
devise an optimal strategy for this process, which means finding an optimal 
input state that achieves the highest possible probability of discrimination.

This issue has already attracted a lot of attention from the scientific
community. It is well established that the discrimination between unitary
operations does not require entanglement~\cite{duan2007entanglement}. In
\cite{cao2016local}, authors have presented a scheme for complete local
discrimination for various kinds of unitary operations. The results in
\cite{cao2015perfect} indicate that it is possible to perfectly distinguish
projective measurements with the help of measurement--unitary
operation--measurement scheme. A single-shot scenario was studied in
\cite{sedlak2014optimal} for $m$ measurements and $n$ outcomes. The authors have
also managed to show that ancilla-assisted discrimination can outperform the
ancilla-free one for perfect distinguishability. The case when the black box can
be used multiple times was investigated by the authors of
\cite{wang2006unambiguous}, who have also proven that the use of entanglement
can improve the discrimination. This issue was also studied in
\cite{d2001using}, where it was shown that entanglement in general improves
quantum measurements for either precision or stability. According to the authors
of \cite{sedlak2009unambiguous}, the optimal strategy for discrimination between
two unknown unitary channels is closely related to problem of discriminating
pure states. They also postulate that entanglement is a key factor in designing
an optimal experiment for a comparison. In the work of A. Jen{\v{c}}ov{\'a} and
M. Pl{\'a}vala, \cite{jenvcova2016conditions}, the optimality conditions for
testers in distinguishability of quantum channels were obtained by the use of
semidefinite programming. The optimal strategies with the use of either
entangled or not entangled  states for the discrimination of Pauli channels were
compared by M. Sacchi in \cite{sacchi2005optimal}. %Another idea of quantifying
% quantum coherence, which utilizes entanglement, was developed by Streltsov
% \etal  in \cite{streltsov2015measuring}

In this work we study the problem of discriminating von Neumann positive
operator valued measures (POVMs). We associate a POVM with a quantum channel and
study the distinguishability of these channels. These channels output classical
probability distributions, hence  we first apply known results for
distinguishing classical probability distributions. The results are applicable
for the case when we are not able to utilize entangled states to perform
discrimination. This, somewhat limited, approach gives us a good starting point
towards our main result. We obtain that entanglement-assisted discrimination of
von Neumann POVMs is related to the discrimination of unitary channels. This
allows us to find a simple condition for perfect discrimination of measurements.
Additionally, we are able to write this result as a semidefinite program (SDP)
which is numerically efficient. The problem gets more complex in the case when
the probability of correct discrimination is strictly less then one. In this
case we have a convex program which calculates the maximum probability of
correct discrimination. Furthermore, it gives us the optimal input state for
this case.

This paper is organized as follows. In Section~\ref{sec:formulation} we
formulate our problem by introducing necessary concepts concerning
discrimination of measurements with and without the assistance of entanglement.
Mathematical framework necessary for stating our results is introduced in
Section~\ref{sec:mathematical-framework}. In
Section~\ref{sec:discrimination-without-entanglement} we consider the case of
discrimination without the assistance of entanglement and provide a necessary
and sufficient criterion for perfect discrimination of two von Neumann
measurements in this case. Entanglement-assisted discrimination of two von
Neumann measurements is analyzed in
Section~\ref{sec:entanglement-assisted-discrimination}. In this section we state
an exact expression for the optimal probability of correct distinction of two
measurements and relate it to the discrimination probability of unitary
channels. We provide a necessary and sufficient condition for perfect
discrimination of two von Neumann measurements as well as a semidefinite program
which is able to check this condition. We also state a simple necessary and a
simple sufficient conditions for perfect discrimination. Finally, we formulate a
convex program which provides the optimal input state for discrimination of two
von Neumann measurements. In Section~\ref{sec:special-cases} we analyze special
cases, that is we consider the discrimination problem of measurement in the
Fourier basis of any dimension and a measurement in the computational basis. We
derive the optimal input state for this task and identify the cases when
entanglement is (not) necessary. Similarly, we consider mirror isometries and
provide a full description of this case. Concluding remarks are presented in the
final Section~\ref{sec:final-remarks}, while proofs of main theorems are
relegated to Appendix~\ref{sec:appendix}.

%%%%%%%%%%%%%%%%%%%%%%%%%%%%%%%%%%%%%%%%%%%%%%%%%%%%%%%%%%%%%%%%%%%%%%%%%%%%%%%%
\section{Formulation of the problem}\label{sec:formulation}
%%%%%%%%%%%%%%%%%%%%%%%%%%%%%%%%%%%%%%%%%%%%%%%%%%%%%%%%%%%%%%%%%%%%%%%%%%%%%%%%

Consider the following scenario. There is an unknown measurement device and the
only thing we know about it is that it performs one of two known measurements,
call them $\mathcal{S}$ and $\mathcal{T}$. We put a state into the device and
our goal is to decide which of the measurements is performed. We aim to identify
the assumptions needed for perfect discrimination of $\mathcal{S}$ and $\TT$.
Further, we want to construct the optimal discrimination scheme for this task. 
In
the case when perfect distinctions is not possible, we would like to bound
from above the probability of correct discrimination as well as derive a scheme
which enables a correct guess with the optimal probability.

The second field of our interest is finding the optimal strategy for the
discrimination. In other words, we would like to know which state should be used
to provide the greatest possible probability of correct discrimination.

In the simplest approach we may think of measurements $\mathcal{S}$ and $\TT$
as measure-and-prepare channels outputting diagonal states, that is classical
probability distributions, see Fig.~\ref{fig:non-ent-disc}. This notion will be
formalized in later sections. Thus, the simplest approach to this problem is to
consider the distance between probability distributions. We can use the distance
between these distributions as an upper bound on the probability of correct
discrimination. In this setting it is also straightforward to find the optimal
state for the discrimination.

Of course, there is another possibility. As we are dealing with quantum states,
we can utilize entanglement. Hence, we input one part of the entangled state
into the unknown measurement device and later use the other part to strengthen
the inference. The scheme of this process is presented in
Fig.~\ref{fig:diamond}.

\begin{figure}[h]
\centering
\includegraphics[width=0.7\linewidth]{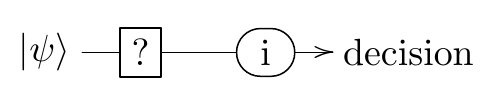}
\caption{
A schematic representation of the setup for distinguishing 
measurements without entanglement.}
\label{fig:non-ent-disc}
\end{figure}

\begin{figure}[h]
\centering 
\includegraphics[width=0.9\linewidth]{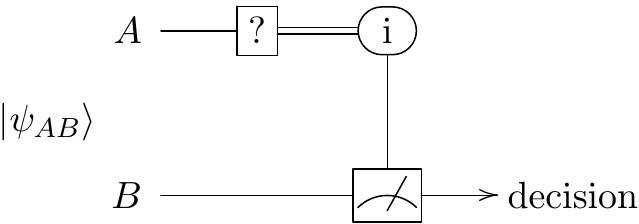} 

\caption{ A schematic representation of the setup for distinguishing
measurements using entangled states. One of two known measurements $\mathcal{S}$
or $\TT$ is performed on part $A$ of the input state $\ket{\psi_{AB}}$. We use
the output label $i$ and perform a conditional binary measurement
$\mathcal{R}_i$ on part $B$. By the use of its output we formulate our guess, 
that is we decide  
weather
the measurement performed on part $A$ was $\mathcal{S}$ or $\TT$.
}\label{fig:diamond}
\end{figure}

%%%%%%%%%%%%%%%%%%%%%%%%%%%%%%%%%%%%%%%%%%%%%%%%%%%%%%%%%%%%%%%%%%%%%%%%%%%%%%%%
\section{Mathematical framework}\label{sec:mathematical-framework}
%%%%%%%%%%%%%%%%%%%%%%%%%%%%%%%%%%%%%%%%%%%%%%%%%%%%%%%%%%%%%%%%%%%%%%%%%%%%%%%%
Let us introduce the following notation. We denote the
matrices of dimension $d_1 \times d_2 $ over the field $\C$ as $M_{d_1,d_2}$. To
simplify, square matrices will be denoted $M_{d}$. The subset of $M_d$
consisting of Hermitian matrices of dimension $d$ will be denoted by $\HH_d$,
while the set of positive semidefinite matrices of dimension $d$ by $\HH_d^+$.
The set of quantum states $\rho$, that is positive semidefinite operators of 
dimension
$d$ such that $\Tr \rho = 1$, will be denoted $\Omega_d$. The set of unitary
matrices of size $d$ will be denoted by $\UU_d$, and its subset
of diagonal unitary matrices of dimension $d$ will be denoted by $\diaguni_d$. 
We 
will also need a linear  mapping transforming $M_{d_1}$ into $M_{d_2}$, which 
will be denoted
\begin{equation}
\Phi : M_{d_1} \to M_{d_2}.
\end{equation}
Finally, we introduce a special subset of all mappings $\Phi$, called quantum 
channels, which are completely positive and trace 
preserving. In other words, the first condition reads
\begin{equation}
\forall A \in \HH^+_{d_1^2} \;\;\; (\Phi \otimes \1)(A) \in \HH_{d_2 d_1}^+,
\end{equation}
while the second one implies 
\begin{equation}
\forall X \in M_{d_1} \tr\Phi(X) = \tr(X).
\end{equation}

The most general form of describing quantum measurements utilizes the notion of
positive operator valued measures (POVMs). In this case a measurement $\TT$ is
given by a set of positive operators $\left\{ T_1, \ldots, T_n \right\}$, for
which we impose the condition  $\sum_i T_i = \1$. Each $T_i \in \HH_d^+$ is
called an effect associated with the label $i$.

While performing a measurement on some quantum state $\rho \in \Omega_d$, the
probabilities of obtaining each of the outcomes $i$ are $p_i=\tr \rho T_i$. Such
measurements can be considered as measure-and-prepare channels. The action of a
channel $\TT$ is given by
\begin{equation}
\TT(\rho) = \sum_{i=1}^n p_i \proj{i}.
\end{equation}

We will be interested in projective rank-one measurements. In this case we have
$n=d$. We will denote such measurements as $\PP_U$. Here $U\in \UU_d$ and the
effects are $P_i = \proj{u_i}$, where $\ket{u_i} = U\ket{i}$, i.e. the
$i$\textsuperscript{th} column of $U$. We arrive at
\begin{equation}
\PP_{U}(\rho) = \sum_{i=1}^d \bra{u_i} \rho \ket{u_i} \proj{i}.
\end{equation}

Now we introduce the bijection between linear operators and vectors in the form 
of the vectorization operation $|X\rangle\rangle$. It is defined for base 
operators as $|(\ketbra{i}{j})\rangle\rangle=\ket{i}\ket{j}$ and uniquely 
extended from linearity. We also recall the well-known equality
\begin{equation}
(A \otimes B) |X\rangle\rangle = |AXB^\top\rangle\rangle,
\end{equation}
where $A \in M_{d_1,d_2}$, $B \in M_{d_3,d_4}$ and $X \in M_{d_3,d_1}$.
For any square matrix $C$ we denote by $\diag(C)$ the linear operation which
gives the diagonal of the matrix $C$ and its conjugate operation 
$\diag^\dagger(v)$
which gives a square diagonal matrix with vector $v$ on the diagonal.

Let us now consider linear mappings transforming  square matrices into square 
matrices 
\ie\ $\Phi: M_{d_1} \to M_{d_2}$. It is well known that quantum channels are a 
special subclass of such mappings. 
\begin{definition}
Consider $\Phi: M_{d_1} \to M_{d_2}$. We define its completely bounded trace
norm, also known as a diamond norm, as
\begin{equation}
\|\Phi\|_\diamond = \max_{\|X\|_1=1} \| \left(\Phi \otimes \1\right) (X) \|_1.
\end{equation}
\end{definition}
It can be shown~\cite{watrous}, that for Hermiticity-preserving $\Phi$ we may 
restrict maximization to rank-1 orthogonal projectors of the form 
$\ketbra{x}{x}$.

There exists a linear bijection between linear mappings $\Phi: M_{d_1} \to
M_{d_2}$ and matrices $M_{d_1d_2}$ which was discovered by
Choi~\cite{choi1975completely} and
Jamio{\l}kowski~\cite{jamiolkowski1972linear}. The operator corresponding to
quantum channel $\Phi$, denoted $J(\Phi)$, can be explicitly obtained as
\begin{equation}
J(\Phi) = \sum_{i,j=1}^{d_1} \Phi(\ketbra{i}{j}) \otimes \ketbra{i}{j}.
\end{equation}

It has the following properties:

\begin{enumerate}
\item $\Phi$ is Hermiticity-preserving if and only if $J(\Phi) \in \HH_{d_1 
d_2}$,\label{it:hp}

\item $\Phi$ is completely positive if and only if $J(\Phi) \in \HH_{d_1 
d_2}^+$,\label{it:cp}

\item $\Phi$ is trace-preserving if and only if $\Tr_1 J(\Phi) = 
\1$.\label{it:tp}
\end{enumerate}
From these properties it follows that every completely positive $\Phi$ is
necessarily Hermiticity-preserving. Moreover, the difference of completely 
positive 
mappings is Hermiticity-preserving. Finally, $\Phi$ is a quantum channel if 
it has properties~\ref{it:cp} and~\ref{it:tp}.

Note that in case of a measurement $\TT$,
$J(\TT)$ takes the form of a block diagonal matrix with transposed effects on
the diagonal, that is $J(\TT) = \sum_{i=1}^d \proj{i} \otimes T_i^\top$.

For Hermiticity preserving $\Phi$, we have the following well-known bounds for 
the
diamond norm~\cite{nechita2016almost,jenvcova2016conditions,watrous}

\begin{equation}
\frac{1}{d_1} \| J(\Phi) \|_1 \leq \| \Phi \|_\diamond \leq \|\Tr_1 |J(\Phi)| 
\|.\label{old_ineq}
\end{equation}

The celebrated result by Helstrom~\cite{helstrom1976quantum} gives an upper
bound on the probability of correct distinction between two quantum channels
$\Phi$ and $\Psi$ in terms of their distance with the use of the diamond norm
\begin{equation}
p \leq \frac12 + \frac14 \| \Phi - \Psi \|_\diamond.
\end{equation}
The above inequality can be saturated with an appropriate choice of measurements
on the output space.

Furthermore, for Hermiticity-preserving $\Phi$, we have the following 
alternative formula for the diamond norm~\cite{chiribella2009theoretical,bisio2011quantum,watrous}

\begin{equation}\label{eqn:diamond-sqrt}
\|  \Phi \|_\diamond = \max \{\left\|(\1\otimes \sqrt{\rho}) J(\Phi) (\1\otimes 
\sqrt{\rho})\right\|_1 : \rho \in \Omega_{d_1}\}.
\end{equation}
The state $\rho$, for which $\|  \Phi \|_\diamond=\left\|(\1\otimes 
\sqrt{\rho}) 
J(\Phi) (\1\otimes 
\sqrt{\rho})\right\|_1$, will be called a \textit{discriminator}.

To complete the mathematical introduction let us recall the definition of total 
variational distance between the probability vectors.

\begin{definition}
Given two discrete probability distributions, represented by vectors $p,q \in 
\R^d$, their total variation distance is defined as
\begin{equation}
\| p-q \|_1= \sum_{i=1}^{d} |p_i-q_i| = 2 \max_{\Delta \subseteq 
\{1,\ldots,d\}} \left( \sum_{a \in \Delta} p_a - q_a \right).
\end{equation}
\end{definition}
%%%%%%%%%%%%%%%%%%%%%%%%%%%%%%%%%%%%%%%%%%%%%%%%%%%%%%%%%%%%%%%%%%%%%%%%%%%%%%%%
\section{Discrimination without 
entanglement}\label{sec:discrimination-without-entanglement}
%%%%%%%%%%%%%%%%%%%%%%%%%%%%%%%%%%%%%%%%%%%%%%%%%%%%%%%%%%%%%%%%%%%%%%%%%%%%%%%%
\subsection{Discrimination of classical probability distributions}
Let us consider a simple approach to the discrimination of measurements. The
idea is to distinguish discrete random variables with distributions given by
probability vectors obtained after performing the measurements on some state
$\rho$. The following proposition states the upper bound for correct
discrimination between two measurements in the case we do not use entanglement.

\begin{proposition}\label{lem:no-ent-dist}
Let $\mathcal{S},\mathcal{T}$ be two measure-and-prepare channels with effects 
$\{S_i\}_{i=1}^n$ and $\{T_i\}_{i=1}^n$ respectively. It holds 
that the probability $p$ of their correct discrimination, without the usage of 
entangled states, is upper bounded by the value 
\begin{equation}
\begin{split}
p \leq 
&\frac{1}{2}+\frac{1}{4}\max_\rho 
\| \diag \left[ (\mathcal{S}-\mathcal{T}) (\rho) \right] \|_1 \\
&=\frac{1}{2}+\frac{1}{2}\underset{\Delta \subseteq \{ 1,\ldots ,d \}}{\max}
\left\| \sum_{i \in \Delta} (S_i - T_i) \right\|.
\end{split}
\end{equation}
\end{proposition}
\begin{proof}
We can note that
\begin{equation}\label{eqn:no-ent-dist}
\begin{split}
&\max_{\rho} \| \diag \left[ (\mathcal{S}-\mathcal{T}) (\rho) \right] \|_1 \\ &=
\underset{\rho}{\max} \sum_i 
\left\vert 
\textrm{Tr} 
\left( \rho ( S_i- T_i ) \right) 
\right\vert 
=
\underset{\psi}{\max} \sum_i 
\left\vert 
\bra{\psi}
\left( S_i- T_i \right) \ket{\psi}
\right\vert \\
&=2 \underset{\Delta \subseteq \{ 1,\ldots ,d \}}{\max} 
\underset{\ket{\psi}}{\max}  
\bra{\psi} \left( \sum_{i \in \Delta} \left(  S_i- T_i  \right) 
\right) 
\ket{\psi} \\
&=2 \underset{\Delta \subseteq \{ 1,\ldots ,d \}}{\max} \left\| 
\sum_{i \in \Delta} ( S_i- T_i) \right\|.
\end{split}
\end{equation}
\end{proof}

In the case of projective measurements $\PP_V$ and $\PP_U$, without loss of 
generality, we assume that one measurement can be performed in the 
computational basis, 
i.e. $V=\1$. We have the following fact
\begin{corollary}\label{lem:proj-no-ent-dist}
Let $\PP_{\1}$ and $\PP_U$ be two projective measurements such that $U \in
\UU_d$ for arbitrary $d$. Then the bound from Proposition~\ref{lem:no-ent-dist} 
reads 
\begin{equation}
\begin{split}
p\leq 
&\frac{1}{2}+\frac{1}{2}\underset{\Delta \subseteq \{ 1,\ldots ,d \}}{\max}
\left\| \sum_{i \in \Delta} (\proj{i} - \proj{u_i}) \right\|
\\
&=
\frac{1}{2}+\frac12\sqrt{1-\min_{\Delta \subseteq \{ 1,\ldots ,d \} 
}\sigma_{\min}^2(U_{\Delta})},
\end{split}
\end{equation}
where $\sigma_{\min}$ denotes the minimal singular value and $U_\Delta = 
\{U_{ij}\}_{ij \in \Delta}$ is a principal submatrix of matrix $U$, with rows 
and columns taken from the subset $\Delta$. 
\end{corollary}
\begin{proof}
Proof follows from Proposition~\ref{lem:no-ent-dist} and the result concerning
singular values of the difference of projectors~\cite{wedin1983angles}.
\end{proof}

\begin{remark}\label{rem:rank-deficient}
From the above Corollary we see that $\PP_{\1}$ and $\PP_U$ are perfectly 
distinguishable without entanglement if and only if there exists a 
rank-deficient principal submatrix of matrix $U$.
\end{remark}

\begin{remark}[Optimal strategy for discrimination of measurements 
without entanglement]\label{rem:strategy-classical}
The optimal input state is the normalized leading eigenvector ($ev_1(\cdot)$) 
of the matrix $\left| \sum_{i \in \Delta} \left(  S_i- T_i  \right) \right|$, 
i.e.
\begin{equation}
\ket{\psi_{\mathrm{opt}}} 
= ev_1 \left(
\left| \sum_{i \in \Delta} \left(  S_i- T_i  \right) \right|
\right)
\end{equation}
for a subset $\Delta$ which maximizes eq.~\eqref{eqn:no-ent-dist}.
In the case of projective measurements it reads
\begin{equation}
\ket{\psi_{\mathrm{opt}}} =
 ev_1 \left(
\left| \sum_{i \in \Delta} (\proj{i} - \proj{u_i}) \right|
\right).
\end{equation}
\end{remark}

\subsection{Discrimination of unitary channels}

Before we proceed to presenting our main results, we need to briefly discuss the
problem of discrimination of unitary channels. This can be done without the
usage of entangled input. In order to formulate the condition for perfect
discrimination of unitary channels we introduce the notion of numerical range of
a matrix $A \in M_d$, denoted by $W(A) =\{\bra{x}A\ket{x}: \ket{x} \in \C^d, \;
\;\braket{x}{x}=1\}$. The celebrated Hausdorf-T\"oplitz
theorem~\cite{hausdorff1919wertvorrat,toeplitz1918algebraische} states that
$W(A)$ is a convex set and therefore $W(A) =\{\tr A \sigma : \sigma \in \Omega_d
\}$. Let us now recall the well-known~\cite{watrous} result for the
distinguishability of unitary channels.
\begin{proposition}\label{th:unit_channel}
	Let $U \in \UU_d$ and $\Phi_U: \rho \mapsto U \rho U^\dagger$ be a unitary 
	channel. 
	Then 
	\begin{equation}
	\| \Phi_U  - \Phi_{\1} \|_\diamond = 2 \sqrt{1-\nu^2},
	\end{equation}
	where $\nu = \min \left\{|x| : x \in W(U^\dagger) \right\}$. 
\end{proposition}
From the above proposition it follows that unitary channels $\Phi_U, \Phi_{\1}$ 
are
perfectly distinguishable if and only if $0 \in W(U^\dagger)$. The above can
also be formulated as: there exists a density matrix $\sigma$, such that $\tr
U^\dagger \sigma = 0$.

%%%%%%%%%%%%%%%%%%%%%%%%%%%%%%%%%%%%%%%%%%%%%%%%%%%%%%%%%%%%%%%%%%%%%%%%%%%%%%%%
\section{Entanglement assisted
discrimination}\label{sec:entanglement-assisted-discrimination}
%%%%%%%%%%%%%%%%%%%%%%%%%%%%%%%%%%%%%%%%%%%%%%%%%%%%%%%%%%%%%%%%%%%%%%%%%%%%%%%%
A more sophisticated idea for discriminating quantum measurements requires the
use of an entangled state. We put one part of the state into the measurement
device and later use the other part to improve the probability of correct
discrimination. Our goal is to show how the discrimination of projective
measurements is connected with the problem of discrimination of unitary
channels. Finally, we would like to state the analytical form of the optimal
discriminator $\rho$.

The following theorem gives us a simple condition that lets us decide whether
$\PP_U$ and $\PP_\1$ are perfectly distinguishable. This condition is one of the
main results of our work and %We state it as the following theorem.
its proof is postponed to Appendix~\ref{sec:proof-diamond-equivalence}. This due
to the fact, that the proof requires a quite large framework of supporting
lemmas.

\begin{theorem}\label{th:ProjPOVM=minUnitary}
Let $U,V\in \UU_d$ and let $\PP_U$ and $\PP_V$ be two projective measurements. 
Let also 
$\diaguni_d$ be the set of diagonal unitary matrices of dimension $d$. Then
\begin{equation}
\|\PP_U - \PP_V\|_\diamond = \min_{E \in \diaguni_d} \|\Phi_{UE} - 
\Phi_V\|_\diamond,
\end{equation}
where $\Phi_U$ is unitary channel.
\end{theorem}

%\begin{remark}
%Due to the invariance of the diamond norm it is enough to consider the case
%where one measurement is in the computational basis. This approach will be used
%throughout the proofs.
%\end{remark}
%\begin{remark}
%For any $E \in \diaguni$ we have the following upper bound of the diamond norm
%\begin{equation}
%\|\PP_U - \PP_{\1}\|_\diamond \leq \|\Phi_{UE} - \Phi_{\1} \|_\diamond.
%\end{equation}
%\end{remark}
%The proof of Theorem~\ref{th:ProjPOVM=minUnitary} is postponed to
%Appendix~\ref{sec:proof-alt}.
% and it is based on the results mentioned below.

Theorem \ref{th:ProjPOVM=minUnitary} gives us a potentially easy method to
calculate the diamond norm. A simple observation is that if we build projections
$U\proj{i}U^\dagger$ from unitary matrix $U$, then the same projections will be
built from matrix $UE$, where $E \in \diaguni_d$. It means that matrices $UE$ 
form an
equivalence class of matrix $U$. The interesting thing is that a
``properly-rotated'' matrix gives us an easy way of calculating the value of the
diamond norm $\|\PP_U - \PP_\1\|_\diamond$ - it is enough to utilize
Proposition~\ref{th:unit_channel}.
Since all unitary  channels of the form $\Phi_{UE}$ are coherifications of 
channel $\PP_U$~\cite{korzekwa2018coherifying}, the above theorem gives us that 
the value of completely bounded trace norm is the minimal value of the norm on 
the  difference between coherified channels.

\begin{figure}[h]
\centering 
\includegraphics[width=0.8\linewidth]{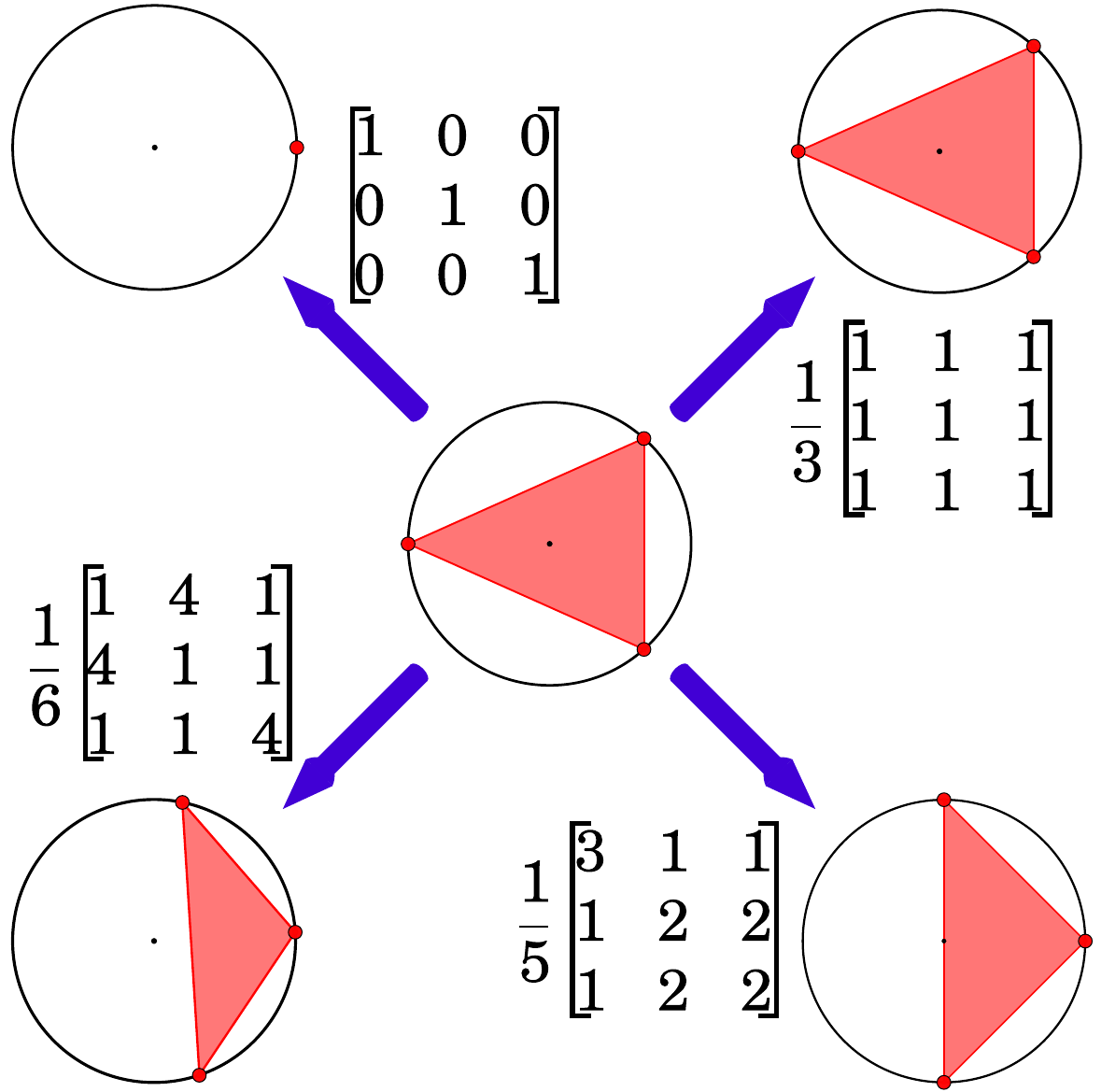} 

\caption{Dependence of the behavior of the numerical range of a matrix $UE \in
\UU_3$ on the eigenvectors of $U$. We start with a matrix $U$ with fixed
eigenvalues and assign each of them distinct eigenvectors. The matrices above 
the
arrows are the unistochastic matrices corresponding to these eigenvector
matrices. The red shaded area is the numerical range of the matrix $UE$ for
which $\min_{E \in \diaguni_d} \|\Phi_{UE} - \Phi_V\|_\diamond$ is achieved.}
\label{fig:range}
\end{figure}

The case of perfect distiguishability can be formulated, by the use of 
Theorem~\ref{th:ProjPOVM=minUnitary}, as a corollary which proof is postponed 
to Appendix~\ref{sec:proof-corollary-diamond-equivalence}.

\begin{corollary}\label{th:diamond-equivalence}
Let $U\in \UU_d$. Then $\PP_U$ and $\PP_{\1}$ are perfectly distinguishable if
and only if for all $E \in \diaguni_d$,  unitary channels $\Phi_{UE}$ are
perfectly distinguishable from the identity channel $\Phi_{\1}$.
\end{corollary}

The above condition together with Proposition~\ref{th:unit_channel} gives us 
that prefect distinguishability is equivalent to the fact that
$\forall_{E\in \diaguni_d} \exists_{\rho}: \tr E^\dagger U^\dagger \rho = 0$. 
In fact, the above is equivalent to
$\exists_{\rho} \forall_{E\in \diaguni_d} : \tr E^\dagger U^\dagger \rho = 0$,
which at first glance seems to be much stronger.
Of course, the latter statement can be rewritten as 
$\exists_{\rho} :  \diag(U^\dagger \rho)= 0$. 
We state this algebraic condition for perfect distinguishability 
in the next theorem, which proof is postponed to
Appendix~\ref{sec:proof-diamond-2}.

\begin{proposition}\label{th:diamond-2}
Let $U\in \UU_d$. Then $\PP_{U}$ and $\PP_{\1}$ are perfectly distinguishable if
and only if there exists $\rho \in \Omega_d$ such that
\begin{equation}\label{eq:diag}
\diag (U^\dagger \rho)= 0 .
\end{equation}
\end{proposition}

We would like to perfectly discriminate the measurements with the lowest
possible amount of entanglement. This translates into the lowest possible rank
of $\rho$. This is shown in the following proposition.
\begin{proposition}
Let $\Phi$ be Hermiticity-preserving and $\rho \in \Omega_{d_1d_2}$ be a 
discriminator of $\Phi$ 
such 
that $\mathrm{rank}(\rho) = k$. Then it is possible to obtain the value of the 
diamond norm on a channel extended by a $k$-dimensional identity channel. If 
the state $\rho$ is rank-one,
then the optimal discrimination can be performed without the use of
entanglement.
\end{proposition}

\begin{proof}
Let us take the Schmidt decomposition of $| \sqrt{\rho}^{\top} \rangle\rangle 
$, that is 
 $| \sqrt{\rho}^{\top} \rangle\rangle =
\sum_{i=1}^{k} \sqrt{\lambda_i} \ket{e_i} \otimes \ket{f_i}$. Then 
\begin{equation}
\begin{split}
\| \Phi \|_\diamond &=\| (\1 \otimes \sqrt{\rho} ) J (\Phi) (\1_d \otimes
\sqrt{\rho} ) \|_1 \\ &= \left\| (\Phi \otimes \1_d )\left(
\ket{{\sqrt{\rho}}^\top}\rangle \langle\bra{{\sqrt{\rho}}^\top} \right) 
\right\|_1
\\ &=\left\|  (\Phi \otimes \1_d )\left( (\1_d \otimes V )
\ket{{\sqrt{\rho}}^\top}\rangle \langle\bra{{\sqrt{\rho}}^\top} (\1_d \otimes V
)^\dagger \right)     \right\|_1
\end{split}
\end{equation}
where $V$ is a unitary matrix such that for the Schmidt decomposition of 
$| \sqrt{\rho}^{\top} \rangle\rangle$ we have
$(\1_d \otimes V )   \ket{{\sqrt{\rho}}^\top}\rangle = \sum_{i=1}^{k}
\sqrt{\lambda_i} \ket{e_i} \otimes \ket{i} $. Thus $(\Phi \otimes \1_d) (|
\sqrt{\rho}^\top \rangle\rangle \langle \langle \sqrt{\rho}^\top |) $ admits a
block structure. Neglecting all zeros we can obtain the same value of the trace
norm for a pure state with the second subsystem of dimension $k$.
\end{proof}
We are especially interested in the case when $\rho$ is a
one-dimensional projection, so we do not need to use entanglement, see
Remark~\ref{rem:rank-deficient} for necessary and sufficient condition in terms
of matrix $U$.

In the general case, the diamond norm of a Hermiticity-preserving $\Phi: M_{d_1}
\to M_{d_2}$ can be computed using the Semidefinite
Program~\ref{eq:sdp}~(from~\cite{watrous2013simpler}).

\begin{table}[!h]
\begin{minipage}[t]{2.1in}
\centerline{\underline{Primal problem}}\vspace{-4mm}
\begin{equation*}
\begin{split}
\text{maximize:}\quad &
\Tr X J(\Phi)
\\[2mm]
\text{subject to:}\quad & 
\begin{bmatrix}
I_{d_2} \otimes \rho & X\\
X^* & I_{d_2} \otimes \rho
\end{bmatrix}
\geq 0\\
& \rho\in \HH_{d_1}^+\\
& X \in M_{d_1d_2}(\mathbb C)
\end{split}
\end{equation*}
\end{minipage}
\qquad
\begin{minipage}[t]{2.1in}
\centerline{\underline{Dual problem}}\vspace{-4mm}
\begin{equation*}
\begin{split}
\phantom{(22)} \\
\text{minimize:}\quad & 
\| \operatorname{Tr}_1 Y \|_\infty
\\[2mm]
\text{subject to:}\quad &
\begin{bmatrix}
Y & -J(\Phi)\\
-J(\Phi) & Y
\end{bmatrix}
\geq 0\\
& Y \in \HH_{d_1 d_2}^+.
\end{split}
\end{equation*}
\end{minipage}
\caption{Semidefinite program for calculating diamond 
norm~\cite{watrous2013simpler}.}\label{eq:sdp}
\end{table}

This program allows us to compute the diamond norm for an arbitrary mapping
$\Phi$. Regretfully, it has one major drawback -- very lengthy computations in
practical applications. In theory, the complexity is polynomial in size of the
input matrix $J(\Phi)$ which has the size of $d_1 d_2 \times d_1 d_2$. Due to
this, the computational time and memory usage allow us to calculate the diamond
norm only for $d_1,d_2<10$.

The result stated in Proposition~\ref{th:diamond-2} is in actuality a simple
%semidefinite constraint. Consider the following. Given $U \in \UU_d$ we wish to
check whether $\PP_U$ can be distinguished perfectly from $\PP_{\1}$ and can
also be used to find a state $\rho \in \Omega_d$ for which
$\|\PP_{\1}-\PP_U\|_\diamond=2$. In the standard approach we would need to solve
the semidefinite programming problem stated in Program~\ref{eq:sdp}.

%As discussed earlier, solving this problem numerically is
%difficult, as its dimensionality grows quickly with $d$.
%
%
%On the other hand our problem
%\begin{equation}
%\begin{split}
%\label{eq:pseudo-sdp}
%\textrm{find} \;\; & \rho \\
%\textrm{such that} \;\; & \rho \geq 0 \\
%& \Tr \rho = 1 \\
%& \diag \rho U = 0,
%\end{split}
%\end{equation}
%is numerically tractable.
%
%Both problems are polynomial in the input data size, but out approach allows 
%to 
%reduce its degree. This is easily seen as our input matrices are of size $d$ 
%and in the original problem they have size $d^2$.

To state the condition~\eqref{eq:diag} formally as a semidefinite program we 
first introduce 
the notation
\begin{equation}
\begin{split}
A_0 & = \1 \\
A_i & = U \proj{i} + \proj{i} U^\dagger, \text{ for } i=1,\ldots, d \\
A_{i} & = \ii \left(\proj{i} U ^\dagger - U \proj{i} \right), \text{ for } 
i=d+1,\ldots,2d.
\end{split}
\end{equation}
Hence we arrive at the primal and dual problems presented in 
Program~\ref{eq:simple-sdp}
\begin{table}[h!]
\begin{minipage}[t]{2in}
\centerline{\underline{Primal problem}} \vspace{-4mm}
\begin{equation*}
\begin{split}
\text{maximize:}\quad &
\Tr \rho A_0
\\[2mm]
\text{subject to:}\quad & 
\Tr \rho A_i = 0\\
& \Tr \rho = 1 \\
& \rho\in \HH_{d}^+
\end{split}
\end{equation*}
\end{minipage}
\qquad
\begin{minipage}[t]{2in}
\centerline{\underline{Dual problem}}\vspace{-4mm}
\begin{equation*}
\begin{split}
\text{minimize:}\quad & 
\bra{0} Y \ket{0}
\\[2mm]
\text{subject to:}\quad &
\sum_{i=0}^{2d} A_i Y_{ii} \geq \1\\
& Y \in \HH_{d}.
\end{split}
\end{equation*}
\end{minipage}
\caption{Semidefinite program for checking perfect distinguishability
of von Neumann measurements.}
\label{eq:simple-sdp}
\end{table}

Note here that the maximization target is a trivial functional, as it reads 
$\tr \rho$ and later we constrain it to $\tr \rho =1$. Hence, the problem 
reduces to satisfying the constraints.

From \cite[Theorem 3]{ambrozie2007finding} we know that the primal problem
of Program~\ref{eq:simple-sdp} has no solutions $\rho \geq 0$ if and only if
\begin{equation}
\inf\limits_{(x_0, \ldots , x_{2d}) \in \mathbb{R} ^{2d+1}} \mathrm{e}^{x_0} 
\tr \left( \mathrm{e}^{ \sum_{i=1}^{2d} x_iA_i } \right) -x_0 = -\infty.
\end{equation}
This is equivalent to the condition that there exists a vector $(x_1, \ldots ,
x_{2d} ) \in \mathbb{R}^{2d} $ such that  $\sum_{i=1}^{2d} x_i A_i < 0$. In a
general case, this is a complicated problem and no analytical methods of finding
a solution are known. Nonetheless, there exist various algorithms, such as 
semidefinite programming, which approximate the solution 
\cite{ambrozie2007finding,bakonyi1995maximum}. The above considerations can be 
summarized as a lemma.
\begin{lemma}\label{lemma:perfect<=>no sign}
Let $U\in \UU_d$ and let $\PP_U, \PP_\1$  be POVMs. Then $\PP_U$ and $\PP_\1$
are perfectly distinguishable if and only if for all real vectors $(x_1, 
\ldots , x_{2d} ) \in
\mathbb{R}^{2d} $ we have $0 \in W\left(\sum_{i=1}^{2d} x_i A_i\right)$.
\end{lemma}
\begin{proof}
The lemma follows directly from the fact that the solution of 
primal problem in Program~\ref{eq:simple-sdp}  exists if and only if the real 
span of $A_i$ contains only matrices without a determined sign.
\end{proof}

The above considerations lead us to the following theorem, which proof is
postponed to Appendix~\ref{sec:proof-diamond2<->numerical-range}.
\begin{lemma}\label{th:diamond2<->numerical-range}
Let $U\in \UU_d$ and let $\PP_U, \PP_\1$  be von Neumann's POVMs. Then $\PP_U$ 
and $\PP_\1$ are perfectly distinguishable  if and only if for all diagonal 
matrices $D$ we have $0 \in W\left(UD + D^\dagger U^\dagger \right)$.
\end{lemma}

As the above conditions for perfect discrimination require solving a
semidefinite problem, here we state a simple necessary and a simple sufficient
conditions based only on the absolute values of the diagonal elements of the
unitary matrix $U$. These turn out to be conclusive in the 3-dimensional case.

\begin{theorem}\label{th:trace}
Let $U \in \UU_d$ and $E \in \diaguni_d$ such that $\bra{i} UE \ket{i} \geq 0$.
Then the following holds:
\begin{enumerate}
\item if $\PP_U$ and $\PP_{\1}$ are perfectly distinguishable, then $\Tr (UE ) 
\leq d-2$
\item if $\Tr (UE) \leq 1$, then $\PP_U$ and $\PP_{\1}$ are perfectly 
distinguishable for odd $d\geq 3$.
\end{enumerate}
In particular, if $d=3$, then  $\PP_U$ and $\PP_{\1}$ are perfectly 
distinguishable if and only if $\Tr (UE ) \leq 1$.
\end{theorem}

\begin{proof}
Assume that $\PP_U$ and $\PP_{\1}$ are perfectly distinguishable. This implies
that $0 \in W(UE) $. Consider a set of possible eigenvalues of $UE$  that
maximizes $\Tr (UE)$ . It can be either $\{ \lambda , -\lambda , 1,  \ldots , 1
\}$ or $\{ \lambda_1 , \lambda_2 , \lambda_3 , 1,  \ldots , 1 \}$, where $0 \in
\textrm{conv}(\lambda_1 , \lambda_2 , \lambda_3 )$. To finish this part of the
proof it is enough to note that $\lambda_1 + \lambda_2 + \lambda_3 \leq 1$.

Let now $E \in \diaguni_d$. We note that $\vert \Tr (UE) \vert \leq 1$. Assume
$0 \not\in W(UE)$. Let $d=2k+1$ and $\lambda(UE)=\{\lambda_1, \ldots,
\lambda_d\}$ be a set of eigenvalues written in an angular order. It is enough
to see that
\begin{equation}
1=\left\vert \lambda_{k+1} \right\vert < \left\vert  \sum_{i=k}^{k+2} \lambda_i
\right\vert < \left\vert \sum_{i=k-1}^{k+3} \lambda_i \right\vert < \ldots
<\left\vert \sum_{i=1}^{d} \lambda_i \right\vert,
\end{equation}
where the first inequality comes from the fact that if we consider unit vectors
on a semicircle, then the absolute value of $\lambda_{k+1}$ can only increase
when added to the sum $\lambda_{k}+\lambda_{k+2}$, which cannot be zero as $0
\not\in W(UE)$. Other inequalities follow from similar reasoning. Thus $\vert
\Tr (UE) \vert > 1$, which finishes the proof.
\end{proof}

Now, we are ready to state the convex program for calculating diamond norm of 
the difference of two von Neumann measurements.

Using Proposition~\ref{th:unit_channel}, the value of the diamond norm from 
Theorem~\ref{th:ProjPOVM=minUnitary} can be rewritten as 
\begin{equation}
\begin{split}
\|\PP_U - \PP_\1 \|_{\diamond} 
&= 
\min_{E \in \diaguni_d} \|\Phi_{UE} - \Phi_\1\|_\diamond =\\ 
&= \min_{E \in \diaguni_d} 2 \sqrt{1- \min_{\rho \in \Omega_d} \vert \Tr \rho 
UE 	\vert^2}\\
&= 2 \sqrt{1- \max_{E \in \diaguni_d} \min_{\rho \in \Omega_d} \vert 
	\Tr \rho UE 
	\vert^2}.
\end{split}
\end{equation}
As shown in Lemma~\ref{th: saddle} from
Appendix~\ref{sec:proof-diamond-equivalence}, we may exchange the minimization
with the maximization and obtain
\begin{equation}
\begin{split}
\nu&:=\max_{E \in \diaguni_d} \min_{\rho \in \Omega_d} \vert  \Tr \rho UE \vert
=
\min_{\rho \in \Omega_d} \max_{E \in \diaguni_d} \vert  \Tr \rho UE \vert =\\
&=
\min_{\rho \in \Omega_d} \sum_{i} \vert  \bra{i}\rho U \ket{i} \vert.
\end{split}
\end{equation}
Now we note that values $\bra{i}\rho U \ket{i} = \tr \rho U \ket{i}\bra{i}$ are
the coefficients of a projection, in the Hilber-Schmidt space, of $\rho$ onto a
subspace $\mathcal{L}_U$ spanned by unit orthogonal vectors $\{U
\ket{i}\bra{i}\}_{i}$. Therefore, the value  $\nu$ is a minimal taxicab norm of
a projection of density matrix $\rho$ onto a subspace $\mathcal{L}_U$ calculated
in the basis $\{U \ket{i}\bra{i}\}_{i}$. The simplified geometrical sketch of
this is presented in Fig.~\ref{fig:cone}.

\begin{figure}[h]
\centering 
\includegraphics[width=0.7\linewidth]{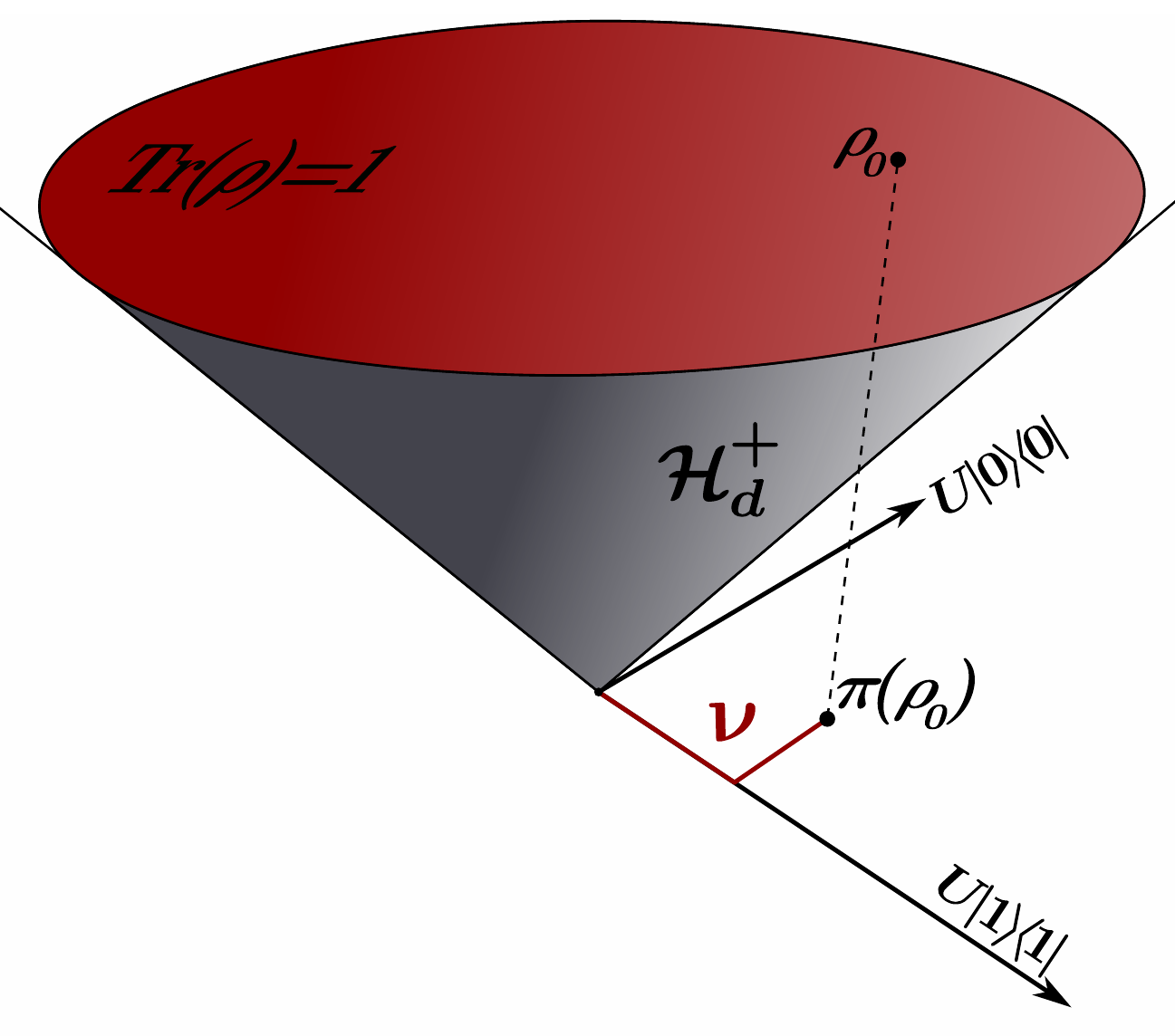} 

\caption{Sketch of the Hilber-Schmidt space with the cone of positive
semidefinite matrices and its intersection with the affine plane $\Tr(\cdot)=1$.
The optimal density matrix $\rho_0$ is marked together with its projection
$\pi(\rho_0)$ onto a plane $\mathcal{L}$ spanned by an orthonormal vectors $\{U
\ket{i}\bra{i}\}_{i}$. The taxicab distance to the origin of the projection
gives the value $\nu$ which in turn determines the diamond norm.}
\label{fig:cone}
\end{figure}

The value $\nu$ can be calculated using cone programming and we get the SDP 
shown in Program~\ref{prog:convex-prog-diamond}.
\begin{table}[!h]
\begin{center}
\begin{minipage}{5cm}
\centering\underline{Primal problem}
\begin{equation*}
\begin{split}
\text{minimize:}\quad &
\| \diag(U^\dagger \rho) \|_1
\\[2mm]
\text{subject to:}\quad & 
\tr \rho = 1, \\
& \rho \geq 0.
\end{split}
\end{equation*}
\end{minipage}
\caption{Convex program for calculation of the diamond norm of the difference 
of two von Neumann measurements.}
\label{prog:convex-prog-diamond}
\end{center}
\end{table}
The minimum value $\nu$ of this program gives us the value of the diamond norm
as
\begin{equation}
\| \PP_U - \PP_\1 \|_\diamond = 2\sqrt{1-\nu^2}.
\end{equation}
We use a state $\rho$ which minimizes the objective function in 
Program~\ref{prog:convex-prog-diamond} to construct the input state for 
discrimination scheme. The input state $\ket{\psi}$ is a purification of $\rho$,
thus its rank is equal to the dimension of the additional subsystem needed for 
optimal 
procedure.
This program is polynomial in the size of the input matrix $U$.
%%%%%%%%%%%%%%%%%%%%%%%%%%%%%%%%%%%%%%%%%%%%%%%%%%%%%%%%%%%%%%%%%%%%%%%%%%%%%%%%
\section{Special cases}\label{sec:special-cases}
%%%%%%%%%%%%%%%%%%%%%%%%%%%%%%%%%%%%%%%%%%%%%%%%%%%%%%%%%%%%%%%%%%%%%%%%%%%%%%%%

In this section we will present several examples of projective measurements,
which can be perfectly distinguished from a measurement in the computational 
basis.

\subsection{Fourier matrices}
First, we consider the Fourier unitary matrices $F_2 \in \UU_2$ and $F_3 \in
\UU_3$. We note that unitary channels $\Phi_{F_2}, \Phi_{F_3}$ are perfectly
distinguishable from the corresponding identity channels. % respectively.
On the other hand, it is not possible to perfectly distinguish $\PP_{F_2},
\PP_{F_3}$ from the $\PP_{\1_2}$ and $\PP_{\1_3}$, which follows from Theorem
\ref{th:trace}.

In the case of higher dimensions, we have the following corollary.
\begin{corollary}\label{cor:fourier}
Let $d \geq 4$ and $F_d \in \UU_d$ be a Fourier matrix. Then $\PP_{F_d}$ is
perfectly distinguishable from $\PP_{\1}$. The perfect discrimination may be
performed with an entangled input state with additional subsystem of dimension 
2. 
%by using a discriminator $\rho$ that $\mathrm{rank}(\rho)\leq2$.
Moreover, if $d = m^2 n$ for $m,n \in \mathbb{N}, \  m > 1$, then the
discrimination can be done without an entangled input, while this is not 
possible 
for prime dimension.
\end{corollary}
\begin{proof}
Define a matrix $X \in M_d$
\begin{equation}
X= \begin{bmatrix}
4\cos{\frac{2 \pi}{d}}&-2 \cos{\frac{2 \pi}{d}}& 0 &\dots & 0 &-2 \cos{\frac{2 
\pi}{d}} \\
-2  \cos{\frac{2 \pi}{d}}& 1 & 0&\dots  & 0&1\\
0& 0 & 0&\dots  & 0 & 0 \\
\vdots& \vdots & \vdots&\ddots &\vdots &\vdots \\
0& 0 & 0&\dots  & 0 & 0 \\
-2 \cos{\frac{2 \pi}{d}}& 1 & 0 & \dots & 0 & 1
\end{bmatrix},\label{eq:state-f}
\end{equation}
which is positive semidefinite for $d\geq 4$.
Direct calculations show that 
\begin{equation}
%\diag(F_d X)= 
\diag(F_d^\dagger X) = 0
\end{equation}
and $\mathrm{rank}(X) \leq 2$ so, as stated in Proposition \ref{th:diamond-2}, 
we 
have perfect distinguishability.
In the case of $d=m^2 n$ we take 
%\begin{equation}
%X = \ketbra{0}{0} + \ketbra{m n}{m n}
%-\ketbra{m n}{0} -\ketbra{0}{m n},
%\end{equation} 
\begin{equation}
X' = (\ket{0}-\ket{mn})(\bra{0}-\bra{mn}),
\end{equation} 
then it holds that  $\diag (F_d^\dagger X')= 0.$ To check that there does not 
exist rank-one perfect discriminator for prime dimension we need to check if 
among principal submatrices of a Fourier matrix there does not exist a 
rank-deficient one (Remark \ref{rem:rank-deficient}). 
The Chebotarev theorem on roots of unity states that such a submatrix does not 
exist, see e.g.~\cite{frenkel2003simple} or Theorem 4 in~\cite{delvaux2008rank}.
\end{proof}

Th optimal input states for discrimination scheme are purifications of matrices 
$X$ in the proof.
%The above considerations lead us to the following corollary.
%
%\begin{corollary}
%Let $ U= U_1 \otimes U_2$ be a unitary matrix. If there exists $\rho_1$ such 
%that $\diag(U_1^\dagger \rho_1)=0  $, then there exists $\rho = \rho_1 
%\otimes \rho_2$ such that $\diag (U^\dagger \rho)=0$.
%\end{corollary}

%\subsection{Reflection}

\subsection{Reflection matrices}
Now, we will consider a unitary matrix given by a mirror isometries.
\begin{corollary}\label{cor:reflection_half}
Let $\UU_d \ni U= \1 - 2 \ketbra{x}{x}$. Then $\PP_U$ is perfectly 
distinguishable
from $\PP_{\1}$ if and only if $\omega=\max_i |x_i|^2 \leq \frac12$. It is
also possible to use a discriminator $\rho \in \Omega_d$ such that
$\mathrm{rank}(\rho) \leq 2$. Moreover, we can find $\rho$ such that
$\mathrm{rank}(\rho)=1$ if and only if
\begin{equation}\label{eq:subset}
\exists_{\Delta \subset \{0,1,\dots,d-1\}}:\sum_{i\in \Delta} |x_i|^2 = \frac12.
\end{equation}
In the case when $\omega>\frac{1}{2}$, we have $\|\PP_U - \PP_\1\|_\diamond 
=2\sqrt{1-4(\omega-\frac12)^2}.$
\end{corollary}

\begin{proof}
If $\omega \leq \frac12$, we provide a construction
\begin{equation}
\rho = \frac12 \ketbra{x}{x} + \frac12 \ketbra{y}{y},
\end{equation}
where
\begin{equation}
y_i  = |x_i| \mathrm{e}^{\ii \alpha_i}
\end{equation}
such that 
\begin{equation}
\braket{y}{x} = 0 = \sum_i |x_i|^2 \mathrm{e}^{\ii (\arg(x_i) - \alpha_i)}.
\end{equation}
By the polygon inequality we know that such phases $\alpha_i$ do exist, and
therefore we receive $\diag(U^\dagger \rho)=0$.

Next, we can note that the existence of a set $\Delta \subset 
\{0,1,\dots,d-1\}$, 
such that 
$\sum_{i\in \Delta} |x_i|^2 = \frac12$, is equivalent to the fact that 
principal 
submatrix 
$V=\{U_{ij}\}_{{i,j\in \Delta}}$ is rank-deficient. Thus, the third statement 
follows from Remark \ref{rem:rank-deficient}.

Now, we assume that $\omega = |x_0|^2 > \frac12$. The case when $\omega=1$ is 
trivial, so we assume $\omega<1$. Let $E'=\1 - 2\proj{0}$. 
Direct calculation gives us
\begin{equation}
\lambda(UE')=\{2|x_0|^2-1 \pm 2|x_0| 
\sqrt{1-|x_0|^2}\ii,\overbrace{1,\ldots,1}^{d-2}\}.
\end{equation}
Eigenvectors corresponding to outlying eigenvalues have the form
\begin{equation}
\ket{\lambda_\pm}= \ket{x}+(-x_0 \pm \frac{x_0}{|x_0|} \sqrt{1-|x_0|^2}\ii) 
\ket{0}
\end{equation}
and from this form we can see that $|\lambda_{+,i}|=|\lambda_{-,i}|$, and 
according to proof of Theorem \ref{th:ProjPOVM=minUnitary} we have
\begin{equation}
\max_{E \in \diaguni} \min_{ \rho \in \Omega_d} | \Tr UE\rho|=\min_{ \rho \in 
\Omega_d} |\Tr UE'\rho|=2|x_0|^2-1.
\end{equation} Utilizing Theorem \ref{th:ProjPOVM=minUnitary} and Proposition 
\ref{th:unit_channel} we obtain
\begin{equation}
\|\PP_U - \PP_\1\|_\diamond=2\sqrt{1-4(\omega-\frac12)^2}.
\end{equation}
\end{proof}

%%%%%%%%%%%%%%%%%%%%%%%%%%%%%%%%%%%%%%%%%%%%%%%%%%%%%%%%%%%%%%%%%%%%%%%%%%%%%%%%
\section{Final remarks}\label{sec:final-remarks}
%%%%%%%%%%%%%%%%%%%%%%%%%%%%%%%%%%%%%%%%%%%%%%%%%%%%%%%%%%%%%%%%%%%%%%%%%%%%%%%%
In this work we have studied the problem of single shot discrimination of two
von~Neumann measurements with finitely many outcomes. Our aim was to design an
optimal strategy for the discrimination in both cases: with and without the
assistance of entanglement. We have parametrized both measurements with a single
unitary matrix $U$ and expressed the results using the properties of $U$. In the
first case, when we do not use entanglement, the optimal probability can be
expressed as a function of minimal singular value of a submatrix of the unitary
matrix $U$, see Corollary~\ref{lem:proj-no-ent-dist}. We have also provided a
construction of an optimal input state which enables performing optimal
discrimination strategy in this scenario. In the second case of
entanglement-assisted discrimination, the optimal probability is a function of
minimal taxicab norm of a projection (in the Hilbert-Schmidt space) of a density
matrix onto a plane spanned by vectors $U\ketbra{i}{i}$, see
Theorem~\ref{th:ProjPOVM=minUnitary} and discussion below. Moreover, we have
provided a convex program for calculating this optimal probability and deriving
the optimal input state for entanglement-assisted discrimination scheme.
Finally, we have considered special cases of Fourier matrices and mirror
isometries.

\appendix

\section{Proof of 
Theorem~\ref{th:ProjPOVM=minUnitary}}\label{sec:proof-diamond-equivalence}

%%%%%%%=====Proof of Theorem th:diamond-equivalence======%%%%%%%%%%
In this section we focus on the proof of the case when $\|\PP_U-\PP_V
\|_\diamond < 2$. The equality is covered by
Corollary~\ref{th:diamond-equivalence}, whose proof is presented in
Appendix~\ref{sec:proof-corollary-diamond-equivalence}.

In order to state the proof of Theorem~\ref{th:ProjPOVM=minUnitary} we will need
the following lemmas. Their proofs are in Appendix~\ref{sec:appendix-lemmas}.

The first lemma states that the distance between von Neumann measurements can be
upper bounded by the distance between unitary channels.
\begin{lemma}\label{lem: inequality}
	Let $U \in \UU_d$ and let $\PP_U$ and $\PP_{\1}$ be two projective 
	measurements. Then for diagonal unitary matrix $E$ of dimension $d$ we have
	\begin{equation}
		\|\PP_U - \PP_{\1}\|_\diamond \leq \|\Phi_{UE} - \Phi_{\1}\|_\diamond,
	\end{equation}
	where $\Phi_U$ is unitary channel.
\end{lemma}

The next lemma states that the optimal point of our optimization problem is in
fact a saddle point. Hence, we may change the order of minimization and
maximization.
\begin{lemma}\label{th: saddle}
	Let $U \in \UU_d$. Then
	\begin{equation}
		\min_{\rho \in \Omega_d} \max_{E \in \diaguni_d} |\Tr(\rho U 
		E)|=\max_{E \in \diaguni_d} \min_{\rho \in \Omega_d} |\Tr(\rho 
		U 
		E)|.
	\end{equation}
\end{lemma}

The third and final lemma tells us about the optimal discriminator.
\begin{lemma}\label{th: state}
Let
\begin{itemize}
\item $E_0 \in \diaguni_d$ and $U \in \UU_d$, $D(E) = \min_{\rho\in \Omega_d} |\Tr \rho UE|$,

\item $D(E_0)>0$,

\item $\lambda_1,\lambda_d$ denote the eigenvalues of $UE_0$ such that the arc
between them is the largest,

\item $P_1$, $P_d$ denote the projectors on the subspaces spanned by the eigenvectors corresponding to $\lambda_1$, $\lambda_d$.
\end{itemize}
Then, the function $|\Tr(\rho UE)|$ has saddle point in $(\rho_0, E_0)$ if and
only if there exist states $\rho_1, \rho_d$ such that
\begin{itemize}
\item $\rho_1=P_1 \rho_1 P_1$,
\item $\rho_d=P_d \rho_d P_d$,
\item $\diag (\rho_1)=\diag (\rho_d)$.
\end{itemize}
\end{lemma}

\subsection{Proof of Theorem~\ref{th:ProjPOVM=minUnitary}}
\begin{proof}[Proof of Theorem~\ref{th:ProjPOVM=minUnitary}]
	For the case when  $\min_{E \in \diaguni_d} \|\Phi_{UE} - 
	\Phi_\1\|_\diamond = 2$, we know that according to 
	Corollary~\ref{th:diamond-equivalence} $\|\PP_U - \PP_{\1}\|_\diamond =2 $ 
	
	Now, we will show the remaining part in the case when 
	$\min_{\diaguni_d}\|\Phi_{UE} - |\Phi_{\1}\|_{\diamond} < 2$.
	%\end{proof}
%%%%%%%=====END=Proof of Theorem th:diamond-equivalence=====%%%%%%%%%%
%\subsection{Proof of Theorem~\ref{th:ProjPOVM=minUnitary}}\label{sec:proof-alt}
%\begin{proof}[Proof of Theorem~\ref{th:ProjPOVM=minUnitary}]
We utilize Lemma~\ref{th: saddle} to obtain existence of saddle point 
$(\rho_0, E_0)$ and Lemma~\ref{th: state} to define new state
\begin{equation}
\tau = \frac12 (\rho_1 + \rho_d )
\end{equation}
and calculate 
$\left\|(\1\otimes 
\sqrt{\tau}) 
J(\PP_\1-\PP_{UE_0}) (\1\otimes 
\sqrt{\tau})\right\|_1$ according to eq. \eqref{eq:igui}. Direct calculation 
gives us
\begin{equation}
\begin{split}
&\sum_{i=1}^{d} \sqrt{\left(\bra{i} \tau \ket{i} + \bra{u_i} \tau 
	\ket{u_i} 
	\right)^2 - 4 \left\vert \bra{i} \tau \ket{u_i} \right\vert^2} =2 
	\sqrt{1-\left\vert \frac{\lambda_1+\lambda_d}{2} \right\vert^2},
\end{split}
\end{equation}
where $\left\vert \frac{\lambda_1+\lambda_d}{2} \right\vert= \vert \Tr 
\tau UE_0 \vert$.
To end this proof we use Lemma~\ref{lem: inequality} and write
\begin{equation}
\begin{split}
&2\sqrt{1-\left\vert \frac{\lambda_1+\lambda_d}{2} 
\right\vert^2}=\left\|(\1\otimes 
\sqrt{\tau}) 
J(\PP_\1-\PP_{U}) (\1\otimes 
\sqrt{\tau})\right\|_1\\
\leq& \|\PP_U - \PP_{\1}\|_\diamond \leq \min_{E \in \diaguni} \|\Phi_{UE} - 
\Phi_\1\|_\diamond=2\sqrt{1-\left\vert \frac{\lambda_1+\lambda_d}{2} 
\right\vert^2}
\end{split}.
\end{equation}
\end{proof}

\subsection{Proofs of Lemmas~\ref{lem: inequality},~\ref{th: saddle} and ~\ref{th: state}}\label{sec:appendix-lemmas}

\begin{proof}[Proof of Lemma~\ref{lem: inequality}]
	Let $\rho^\top$ be a discriminator of $\PP_U - \PP_{\1}$. Thus
	\begin{equation}
	\begin{split}
	\|\PP_U - \PP_{\1}\|_\diamond &=
	\left\|
	\left(\1 \otimes \sqrt{\rho^\top}\right) J_{\PP_U - 
		\PP_{\1}} \left(\1 \otimes \sqrt{\rho^\top}\right)
	\right\|_1\\
	&=
	\left\|\sum_i \proj{i} \otimes \sqrt{\rho^\top}M_i^\top \sqrt{\rho^\top} 
	\right\|_1, 
	\end{split}
	\end{equation}
	where $M_i = \proj{i} - \proj{u_i} = \proj{i} - UE \proj{i} E^\dagger 
	U^\dagger$. Now, using the operational definition of the trace norm 
	($\|A\|_1 = 
	\max_{V \in\UU_d} |\tr 
	AV|$) and the fact that the matrix is in a block form, we obtain
	\begin{equation}
	\begin{split}
	&\left\|\sum_i \proj{i} \otimes \sqrt{\rho^\top}M_i^\top \sqrt{\rho^\top} 
	\right\|_1
	=\sum_i \tr (\sqrt{\rho} M_i \sqrt{\rho} V_i)\\
	&=
	\tr\left( \sum_i \proj{i} \otimes \sqrt{\rho}M_i \sqrt{\rho}\right)
	\left( \sum_i \proj{i}\otimes V_i\right)\\
	\end{split}
	\end{equation}
	where $V_i$ is a unitary matrix, which is optimal for $i^{\text{th}}$ block.
	Next we note that
	\begin{widetext}
		\begin{equation}
		\begin{split}
		&\tr \left( \sum_i \proj{i} \otimes \sqrt{\rho} \left( \proj{i} - UE 
		\proj{i} 
		E^\dagger U^\dagger \right) \sqrt{\rho}\right)
		\left( \sum_i \proj{i}\otimes V_i \right) \\
		=& \tr \left( \sum_{ij}\ketbra{i}{j} \otimes \sqrt{\rho} \left( 
		\ketbra{i}{j} - 
		UE \ketbra{i}{j}E^\dagger U^\dagger \right) \sqrt{\rho}\right)
		\left( \sum_i \proj{i}\otimes V_i\right)\\
		\leq&
		\max_{V\in \UU(d^2)}
		\left|\tr
		\left( \sum_{ij}\ketbra{i}{j} \otimes \sqrt{\rho} \left( \ketbra{i}{j} 
		- 
		UE \ketbra{i}{j} E^\dagger U^\dagger \right) \sqrt{\rho}\right)
		V \right|
		=
		\left\|
		\sum_{ij}\ketbra{i}{j} \otimes \sqrt{\rho} \left( \ketbra{i}{j} - 
		UE \ketbra{i}{j}E^\dagger U^\dagger \right) \sqrt{\rho}
		\right\|_1 \\
		=&
		\left\|
		\left(\1  \otimes \sqrt{\rho}\right) 
		\left(\projV{\1} -\projV{(UE) ^\top}\right)
		\left(\1  \otimes \sqrt{\rho}\right)
		\right\|_1
		\leq \|
		\Phi_{(UE)^\top} - \Phi_{\1}
		\|_\diamond = 
		\|
		\Phi_{UE} - \Phi_{\1}
		\|_\diamond.
		\end{split}
		\end{equation}
	\end{widetext}
\end{proof}

The next proof uses Corollary~\ref{th:diamond-equivalence} and Proposition~\ref{th:diamond-2}. Their proofs are stated in Appendix~\ref{sec:appendix}.
\begin{proof}[Proof of Lemma~\ref{th: saddle}]
	Let $\min_{E \in \diaguni_d} \|\Phi_{UE} - \Phi_\1\|_\diamond =2$. 
	Utilizing Corollary~\ref{th:diamond-equivalence} and 
	Proposition~\ref{th:diamond-2} there exists a state $\rho_0$ such 
	that for each $E \in \diaguni_d$ $|\Tr (\rho_0 UE)|=0$. Fallowing 
	Proposition~\ref{th:unit_channel} we obtain
	\begin{equation}
	\begin{split}
	\max_{E \in \diaguni_d} \min_{\rho \in \Omega_d} |\Tr(\rho 
	U 
	E)|=0=\max_{E \in \diaguni_d} |\Tr(\rho_0 UE)| \\ =\min_{\rho \in \Omega_d} 
	\max_{E \in \diaguni_d} |\Tr(\rho U	E)|.
	\end{split}
	\end{equation}
	
	Now assume $\min_{E \in \diaguni_d} \|\Phi_{UE} - \Phi_\1\|_\diamond <2$. 
	We have
	\begin{equation}
	\begin{split}
	\min_{E \in \diaguni_d} \|\Phi_{UE} - \Phi_\1\|_\diamond 
	&= 2 \sqrt{1- \max_{E \in \diaguni_d} \min_{\rho \in \Omega_d} \vert 
		\Tr \rho UE 
		\vert^2}.
	\end{split}
	\end{equation}	
	In the case of $\rho_0 \in \Omega_d$ and $E_0 \in \diaguni_d$ which 
	saturate 
	$\min_{E 
		\in \diaguni_d} \|\Phi_{UE} - \Phi_\1\|_\diamond$, we have that $0 
		\notin 
	W(UE_0)$. 
	
	Let $\diagmodul$ be the 
	set of diagonal matrices $E$ such that $\vert E_{ii} \vert \leq 1$.
	The set of density matrices and the set $\diagmodul$ are both compact and 
	convex. 
	Moreover, the sets $\{E \in \diagmodul: \Re(\Tr(\rho U E))=\max_{D \in 
		\diagmodul} \Re(\Tr(\rho U D))\}$ and  
	$\{\rho \in \Omega_d: \Re(\Tr(\rho U E))=\min_{\sigma \in \Omega_d} 
	\Re(\Tr(\sigma U E))\}$ are convex. 
	Since all assumptions of the Theorem 3 in \cite{fan1952fixed} are 
	fulfilled, 
	we obtain the existence of saddle points, and therefore 
	\begin{equation}\label{eq:minmaxrealis}
	\min_{\rho \in \Omega_d} \max_{E \in \diagmodul} \Re \left(\Tr(\rho U 
	E)\right)=\max_{E \in \diagmodul} \min_{\rho \in \Omega_d} \Re 
	\left(\Tr(\rho 
	U 
	E)\right).
	\end{equation}
	One can note that it implies that for a saddle point $(\rho_0, E_0)$ we 
	have 
	$\Re 
	(\Tr 
	\rho_0U E_0)=\Tr \rho_0 U E_0 = |\Tr \rho_0 U E_0|$. Moreover, $\max_E |\Tr 
	\rho_0 UE|=\sum_i |\bra{i}\rho_0U\ket{i}|=\Tr \rho_0U E_0$ and $\Tr \rho_0 
	U 
	E_0 = \min_\rho |\Tr \rho U E_0|$. That means $(\rho_0,E_0)$ is the saddle 
	point of $|\Tr \rho U E|$ and 
	\begin{equation}
	\min_{\rho \in \Omega_d} \max_{E \in \diagmodul} |\Tr(\rho U 
	E)|=\max_{E \in \diagmodul} \min_{\rho \in \Omega_d} |\Tr(\rho 
	U 
	E)|.
	\end{equation}
	Let us write $E_0=F_0 D$, where $F_0 \in \diaguni_d$ and $D$ is a diagonal 
	matrix with $0\leq D_{ii} \leq 1$. We will 
	show that we have the saddle point also for $(\rho_0,F_0)$.
	First of all, we will observe that for arbitrary $U \in \UU_d$
	\begin{equation}
	\min_\rho |\Tr \rho U| \geq \min_\rho |\Tr \rho U D|.
	\end{equation}
	For the case when $0 \in W(U)$, for some probability vector $p$ we have 
	$\sum_i 
	\lambda_i p_i=0$, where $\lambda_i$ are the eigenvalues of $U$. If there 
	exists 
	$i$ such that $\bra{\lambda_i}D\ket{\lambda_i}=0$, then $|\Tr 
	\proj{\lambda_i} 
	U D|=0$. Otherwise, we can take the state $\rho=\sum_i q_i 
	\proj{\lambda_i},$ 
	where 
	$q_i=\frac {p_i}{\bra{\lambda_i}D\ket{\lambda_i}}$ and notice that $0 \in 
	W(UD)$. In 
	the case when $0 \not\in W(U)$ for the most distant pair of eigenvalues 
	$\lambda_1, 
	\lambda_d$,  using T\"oplitz-Hausdorff theorem, we have an inclusion of 
	the  
	interval 
	in 
	a numerical range
	\begin{equation}
	\begin{split}
	&[\Tr \proj{\lambda_1}UD, 
	\Tr\proj{\lambda_d}UD]\\
	=&[\lambda_1\bra{\lambda_1}D\ket{\lambda_1},\lambda_d 
	\bra{\lambda_d}D\ket{\lambda_d}] \subset W(UD).
	\end{split}
	\end{equation} 
	In our case using the optimality condition we receive $\min_\rho |\Tr \rho 
	U F_0| = 
	\min_\rho 
	|\Tr \rho U F_0D|$. Now, we are 
	ready to check whether $(\rho_0, F_0)$ is the saddle point. We write
	\begin{equation}
	\begin{split}
	&|\Tr \rho_0 U F_0| \leq \max_{E \in \diagmodul} |\Tr \rho_0UE|=|\Tr \rho_0 
	U 
	E_0| \\
	=& \min _\rho |\Tr \rho U F_0 D| =\min_\rho |\Tr \rho UF_0| \leq |\Tr 
	\rho_0 U 
	F_0|.
	\end{split}
	\end{equation}
	The above gives us information that $$|\Tr \rho_0 U F_0|=\min_\rho |\Tr 
	\rho 
	UF_0|=\max_{E \in \diagmodul} |\Tr \rho_0UE|.$$
	That means
	\begin{equation}
	\min_{\rho \in \Omega_d} \max_{E \in \diaguni_d} |\Tr(\rho U 
	E)|=\max_{E \in \diaguni_d} \min_{\rho \in \Omega_d} |\Tr(\rho 
	U 
	E)|.
	\end{equation}
\end{proof}

\begin{proof}[Proof of Lemma~\ref{th: state}]
First we show the reverse implication. Define $\rho_0= \frac12
(\rho_1+\rho_d)$. We see that $|\Tr(U E_0 \rho_0)|= D(E_0)$.
For arbitrary $E \in \diaguni_d$ direct calculation gives us
\begin{equation}
|\Tr(U E_0 \rho_0)| \geq |\Tr(U E \rho_0)| \geq \min_{\rho \in \Omega_d} 
|\Tr(U E \rho)|
\end{equation}
That means $D(E_0) \geq D(E)$ and 
$|\Tr(U E_0 \rho_0)|=\min\limits_\rho |\Tr(U E_0 \rho)|=\max\limits_E 
|\Tr(U E \rho_0)|$.

Now we prove the direct implication. Without loss of generality we may assume
$\lambda_1=\lambda$ and $\lambda_d=\overline{\lambda}$. Since $\rho_0$ gives
minimum of the $|\tr \rho U E|$, thus $\rho_0$ is supported on the subspace
spanned by the range of $P_1$ and $P_d$, i.e.
\begin{equation}
\rho_0 = P \rho_0 P \text{ for } P= P_1+P_d.
\end{equation}
We may write 
\begin{equation}
\rho_0 = P \rho_0 P = P_1 \rho_0 P_1 + P_d \rho_0 P_d + P_1 \rho P_d + P_d 
\rho_0 P_1
\end{equation}
and define 
\begin{equation}
\begin{split}
\rho_1&= P_1 \rho_0 P_1, \\
\rho_d&= P_d \rho_0 P_d, \\
\rho_{1d}&= P_1 \rho_0 P_d, \\
\rho_{d1}&= P_d \rho_0 P_1.
\end{split}
\end{equation}
Note that the optimality forces  $\tr \rho_1 = \tr \rho_d = \frac12.$
Now we write
\begin{equation}
z_i = \bra{i} \rho_0 U E_0\ket{i} = 
\lambda \bra{i} \rho_1\ket{i} 
+
\overline{\lambda} \bra{i} \rho_d\ket{i} 
+
2 \Re( \lambda \bra{i} \rho_{d1}\ket{i} ).
\end{equation}
We have $\sum_i z_i = \frac{\lambda + \overline{\lambda}}{2}$. If elements $z_i$
have different phases, then by additional diagonal unitary matrix one can
increase the value of the sum and contradict to the fact that $(\rho_0, E_0)$ is
a saddle point. Therefore, we conclude that all elements have the same phase and
therefore we obtain that
\begin{equation}
\bra{i} \rho_1\ket{i}  = \bra{i} \rho_d\ket{i} \text{ for all } i.
\end{equation}
\end{proof}

\section{Proof of Corollary~\ref{th:diamond-equivalence}}\label{sec:appendix}
The proof is based on Proposition~\ref{th:diamond-2} and Lemma~\ref{th:diamond2<->numerical-range}. These proofs are stated later in this appendix.

\subsection{Proof of 
Corollary~\ref{th:diamond-equivalence}\label{sec:proof-corollary-diamond-equivalence}}

\begin{proof}[Proof of Corollary~\ref{th:diamond-equivalence}]
		Let us assume that $\PP_U$ is perfectly distinguishable from 
		$\PP_{\1}$. 
		Then,
		from Proposition~\ref{th:diamond-2} 
		there exists a density matrix such that 
		\begin{equation}
		\diag(U^\dagger \rho) = 0.
		\end{equation}
		Hence, for all $E \in \diaguni_d$ we have $\diag(E^\dagger U^\dagger 
		\rho) = 0$. Therefore $0 \in W(E^\dagger U^\dagger)$, and thus unitary 
		channel  
		$\Phi_{UE}$ is perfectly distinguishable from the identity channel.
		
		Now, we assume that for all $E \in \diaguni_d$ we have $0 \in 
		W(E^\dagger U^\dagger)$. 
		%Therefore $\Phi_{UE}$ is perfectly distinguishable 
		%from the identity channel $\Phi_{\1}$.
		We will show that for any diagonal matrix $D$ (not necessarily 
		unitary), we 
		have $0 \in W\left(UD + D^\dagger U^\dagger \right)$ (see Lemma 
		\ref{th:diamond2<->numerical-range}).
		One may assume that $D$ is invertible as otherwise we would have 
		$\bra{\psi} \left(UD + D^\dagger U^\dagger \right) \ket{\psi} = 0$
		for $\ket{\psi} \in \ker(D)$. We write
		\begin{equation}
		UD = U E D_+,
		\end{equation}
		where $E \in \diaguni_d$ and $D_+$ is a strictly positive diagonal 
		matrix. Let $V$ be a unitary matrix such that
		\begin{equation}
		U E = V \diag^\dagger(\lambda) V^\dagger,
		\end{equation}
		where $\lambda$ denotes eigenvalues of $U E$. From our assumption we 
		have 
		that 
		there 
		exists a probability vector $p$, such that 
		\begin{equation}
		\sum_i \lambda_i p_i =0.
		\end{equation}
		Now we define a density matrix 
		\begin{equation}
		\sigma = V \diag^\dagger(q) V^\dagger,
		\end{equation}
		where
		\begin{equation}
		q_i = c^{-1} \frac{p_i}{
			\bra{i}V^\dagger D_+ V\ket{i}}; 
		\ \  c = {\sum_{j} \frac{p_j}{\bra{j}V^\dagger D_+ V\ket{j}}}.
		\end{equation}
		Using this we obtain
		\begin{equation}
		\tr U D \sigma =
		%\tr V \diag^\dagger({\lambda}) V^\dagger D_+ V  \diag^\dagger(q) 
		%V^\dagger 
		%=
		%\tr \Sigma_{\lambda} V^\dagger D_+ V \Sigma_q 
		c^{-1} \sum_i \lambda_i p_i = 0.
		\end{equation}
		Thus $0 \in W\left(UD\right)$ and therefore $ 0 \in W\left(UD + 
		D^\dagger 
		U^\dagger \right)$.
\end{proof}

\subsection{Proof of Proposition~\ref{th:diamond-2}}\label{sec:proof-diamond-2}
%PROOF OF THEOREM diamond-2 %%%%%%%%%%%%%%%%%%%%%%%%%%%%%%%%%%%%%%%%%%%%%%%%%%%%
\begin{proof}[Proof of Proposition~\ref{th:diamond-2}]
	Let $\rho \in \Omega_d$ be a discriminator. Then
	\begin{equation} \label{eq:igui}
	\begin{split}
	& \| \PP_\1-\PP_U \|_\diamond = 
	\left\| \sum_{i=1}^d \proj{i} \otimes
	\left(\sqrt{\rho} \left( 
	\proj{i} - \proj{u_i} \right) \sqrt{\rho} \right) 
	\right\|_1\\
	&= \sum_{i=1}^d\tr \big\vert \sqrt{\rho} \proj{i} \sqrt{\rho} - 
	\sqrt{\rho}\proj{u_i} \sqrt{\rho} \big\vert \\
	%&=\tr \left\vert \bra{i} \rho \ket{i} \frac{\sqrt{\rho} 
	%\ket{i}}{\sqrt{\bra{i} 
	%\rho \ket{i}}}  \frac{\bra{i} \sqrt{\rho} }{\sqrt{\bra{i} \rho \ket{i}}} - 
	%\bra{u_i} \rho \ket{u_i} \frac{\sqrt{\rho} \ket{u_i}}{\sqrt{\bra{u_i} \rho 
	%\ket{u_i}}}  \frac{\bra{u_i} \sqrt{\rho} }{\sqrt{\bra{u_i} \rho 
	%\ket{u_i}}} 
	%\right\vert \\
	&=\sum_{i=1}^{d} \sqrt{\left(\bra{i} \rho \ket{i} + \bra{u_i} \rho 
	\ket{u_i} 
		\right)^2 - 4 \left\vert \bra{i} \rho \ket{u_i} \right\vert^2},
	\end{split}
	\end{equation}
	where the last equality follows from the singular value decomposition for 
	rank-two 
	matrices.
	
	Assume that $\| \PP_\1 - \PP_U\|_\diamond=2$. If for any  
	state $\rho$, the condition~\eqref{eq:diag} is not satisfied, 
	\ie\ $\forall_\rho \exists_i \  \bra{i}\rho \ket{u_i} \neq 0$, then
	\begin{equation}
	\begin{split}
	\phantom{<}&\sum_{i=1}^{d} \sqrt{\left(\bra{i} \rho \ket{i} + \bra{u_i} 
	\rho \ket{u_i} 
		\right)^2 - 4 \left\vert \bra{i} \rho \ket{u_i} \right\vert^2} \\ 
		<& \sum_{i=1}^d 
	\left( \bra{i} \rho \ket{i} +\bra{u_i} \rho \ket{u_i} \right) =2,
	\end{split}
	\end{equation}
	which gives a contradiction.
	
	Next, assume that there exists a state $\rho$ such that $\bra{i} \rho 
	\ket{u_i} 
	= 0$ for all $i$. From eq. \eqref{eq:igui} we have $\| \PP_\1-\PP_U 
	\|_\diamond 
	=2$.
\end{proof}
%END PROOF OF THEOREM diamond-2 %%%%%%%%%%%%%%%%%%%%%%%%%%%%%%%%%%%%%%%%%%%%%%%%

\subsection{Proof of 
Lemma~\ref{th:diamond2<->numerical-range}}\label{sec:proof-diamond2<->numerical-range}

%PROOF OF th:diamond2<->numerical-range %%%%%%%%%%%%%%%%%%%%%%%%%%%%%%%%%%%%%%%%
\begin{proof}[Proof of Lemma~\ref{th:diamond2<->numerical-range}]
	Perfect distinguishability between $\PP_U$ and $\PP_{\1}$ means, by
	Proposition~\ref{th:diamond-2}, there exists a discriminator $\rho \in 
	\Omega_d$
	we have
	\begin{equation}
	\diag (U^\dagger \rho)= 0.
	\end{equation}
	If this condition is satisfied, we also have $\diag (D^\dagger U^\dagger 
	\rho)= 
	0$ for any diagonal matrix $D$, and therefore $0 \in  W\left(UD + D^\dagger 
	U^\dagger \right)$.
	
	%In fact, if $\PP_U$ 
	%is perfectly distinguishable from $\PP_{\1}$, then $\Phi_{UE}$ is 
	%perfectly 
	%distinguishable from $\Phi_{\1}$ for any diagonal unitary matrix $E$.
	Now, let us assume that for all diagonal matrices $D$ we have $0 \in 
	W\left(UD 
	+ D^\dagger U^\dagger \right)$.
	We define a matrix
	\begin{equation}
	D= \diag^\dagger(x_1 - \ii x_{d+1}, x_2 - \ii x_{d+2}, \dots ,x_d - \ii 
	x_{2d}).
	\end{equation}
	Thus, there exists a nonzero, $x$-dependent state $\ket{\psi}$, such that 
	\begin{equation}
	\bra{\psi} \left(UD + D^\dagger U^\dagger \right)\ket{\psi} = 0.
	\end{equation}
	This can be equivalently expressed as 
	\begin{equation}
	\bra{\psi} \sum x_i A_i \ket{\psi} = 0.
	\end{equation}
	Using Lemma~\ref{lemma:perfect<=>no sign} we arrive at our result.
	%It means that neither $\sum \alpha_i A_i > 0$ nor $\sum \alpha_i A_i < 0$ 
	%for 
	%all $\alpha_i$,  and therefore there exists $\rho$, such that $\diag 
	%(U^\dagger 
	%\rho)= 0.$
\end{proof}
%END PROOF OF th:diamond2<->numerical-range %%%%%%%%%%%%%%%%%%%%%%%%%%%%%%%%%%%%

\section*{Acknowledgements}
This work was supported by the Polish National Science Centre under project 
numbers 2016/22/E/ST6/00062 (ZP, AK, RK) and 2015/18/E/ST2/00327 ({\L}P).
We would like to thank Karol Horodecki for fruitful discussions.

\newpage

\bibliographystyle{apsrev}
\bibliography{measurement_distance}

\end{document}